\documentclass[journal, onecolumn, 12pt]{IEEEtran}
\renewcommand{\baselinestretch}{1.6}

\usepackage{amsmath,amssymb,amscd}
\usepackage{graphicx, cite}

\newtheorem{proposition}{Proposition}

\newtheorem{lemma}{Lemma}
\newtheorem{remark}{Remark}

\def\RppSD{R_{{pp}}^{\rm{SD}}}

\def\RpeSD{R_{{pe}}^{\rm{SD}}}
\def\RpeJD{R_{{pe}}^{\rm{JD}}}

\def\ReMAC{R_{{e,\rm{MAC}}}}

\def\RseSD{R_{{se}}^{\rm{SD}}}
\def\RseJD{R_{{se}}^{\rm{JD}}}

\def\Rss{R_{{ss}}}

\def\Rsec{R_{\rm{secret}}}

\def\Pmax{P_{\max}}
\def\Psmax{P_{s,\max}}

\def\hsp{{\pmb{h}_{sp}}}
\def\hse{{\pmb{h}_{se}}}
\def\hs{{\pmb{h}_{ss}}}
\def\hpp{h_{pp}}
\def\hpe{h_{pe}}
\def\hps{h_{ps}}
\def\sigmap{\sigma_p^2}
\def\sigmae{\sigma_e^2}
\def\sigmas{\sigma_s^2}

\def\wopt{{\pmb{w}_{\rm{opt}}}}

\def\vmax{\pmb{v}_{\max}}

\def\pe{{\pmb{e}}}
\def\pw{{\pmb{w}}}

\def\pZ{{\pmb{Z}}}

\begin{document}
\title{Spectrum Leasing via Cooperation for Enhanced Physical-Layer Secrecy}
\author{Keonkook~Lee,~\IEEEmembership{Member,~IEEE,}
        Chan-Byoung~Chae,~\IEEEmembership{Senior Member,~IEEE},
       Joonhyuk~Kang,~\IEEEmembership{Member,~IEEE}
\thanks{K. Lee and J. Kang are with the Department of Electrical Engineering, Korea Advanced Institute of Science and Technology, Daejeon, Korea (e-mail: klee266@kaist.ac.kr, jhkang@kaist.edu).}
\thanks{C.-B. Chae is with the School of Integrated Technology, Yonsei University, Korea (e-mail: cbchae@yonsei.ac.kr).}
\thanks{Part of this work was presented at the International Conference on Communications (ICC), 2011.}
}

\maketitle

\begin{abstract}
Spectrum leasing via cooperation refers to the possibility of primary users leasing a portion of the spectral resources to secondary users in exchange for cooperation. In the presence of an eavesdropper, this correspondence proposes a novel application of this concept in which the secondary cooperation aims at improving secrecy of the primary network by creating more interference to the eavesdropper than to the primary receiver. To generate the interference in a positive way, this work studies an optimal design of a beamformer at the secondary transmitter with multiple antennas that maximizes a secrecy rate of the primary network while satisfying a required rate for the secondary network. Moreover, we investigate two scenarios depending upon the operation of the eavesdropper: i) the eavesdropper treats the interference by the secondary transmission as an additive noise (single-user decoding) and ii) the eavesdropper tries to decode and remove the secondary signal (joint decoding). Numerical results confirm that, for a wide range of required secondary rate constraints, the proposed spectrum-leasing strategy increases the secrecy rate of the primary network compared to the case of no spectrum leasing.
\end{abstract}


\section{Introduction}\label{sec:intro}

Cognitive radio networks make efficient use of the spectrum by allowing the coexistence of secondary devices in a bandwidth occupied by primary networks \cite{Goldsmith08ProcCR, Sriram08JSTSP, KK11TWC}. Among proposals for the implementation of the cognitive radio, \cite{Osvaldo08JSAC, Tariq11TVT, Igor10TCOM, Davide11TWC} have proposed a spectrum-leasing framework whereby the primary network leases part of the spectral resources to secondary users in exchange for cooperation. Such previous work has considered a scenario where secondary nodes provide cooperation in the form of relaying primary packets in return for the possibility of transmitting their own data in leased spectral resources. The primary system benefits by obtaining achievable rates \cite{Osvaldo08JSAC, Tariq11TVT} or reliability \cite{Igor10TCOM, Davide11TWC}. This correspondence explores an alternative application of the concept of spectrum leasing via cooperation. The main idea is that, if spectrum access is allowed by the primary network, the secondary transmission creates more interference to the eavesdropper than to the primary receiver; thus the secrecy rate of the primary link increases.

The authors in \cite{Wyner75BSTJ, Csiszar78TIT, Bloch08TIT} studied the secrecy capacity in wiretap channel. In an effort to improve the secrecy rate, the authors in \cite{Tang11TIT, Jorswieck10SCC, Goel05VTC, Zhou10TVT} studied an information theoretic analysis of secrecy capacity in the presence of a helper node where the sole role of the helper node was that of increasing the main link's secrecy rate. In this correspondence, we study a design of the secondary transmitter where the secondary transmitter works as a helper node for secrecy rate of the primary network and transmits a message for its own network at the same time. \cite{Wu11TIFS} introduced a similar idea that studies the impact of an interaction between a primary user and a secondary user when all nodes were equipped with a single antenna.\footnote{Note that the our idea of the spectrum leasing via cooperation for secrecy of primary networks was first introduced in our previous work \cite{KK11ICC}.} In our work, we study the case of a multiple antennas and analyze an optimal beamforming vector at the secondary transmitter. The main contributions of this work are: i) the proposal of a spectrum-leasing scheme via cooperation for enhancement of the secrecy rate of the primary network, ii) an analysis of an optimal beamforming vector maximizing a primary secrecy rate while satisfying a required secondary rate. 
In our previous work \cite{KK11ICC}, we have introduced a design of the secondary transmitter when the number of transmit antennas is assumed to be more than three and the eavesdropper treats the interference from the secondary transmitter as an additive noise. In this work, we study two scenarios depending upon the operation of the eavesdropper, i.e., a single-user decoding eavesdropper and a joint decoding eavesdropper. Moreover, we investigate a design of optimal beamforming vectors irrespective of the number of antennas. Finally, we demonstrate that the proposed spectrum-leasing scheme improves the secrecy rate of the primary network for a wide range of secondary rate constraints.

\emph{Notation}:
Lower case and upper case boldface denote vectors and matrices, respectively.
$[\cdot]^*$ denotes conjugate transpose.
$\left\|\pmb{a}\right\|$ denotes the Euclidean norm of $\pmb{a}$.
$[A]^+$ denotes $\max(A,0)$.
$\pmb{z} \sim \mathcal{CN} \left(\pmb{m}, \pmb{V}\right)$ denotes that
the elements of $\pmb{z}$ are random variables with
the distribution of a circularly-symmetric-complex-Gaussian (CSCG)
with mean $\pmb{m}$ and covariance $\pmb{V}$.

\section{System Model}\label{sec_sm}

The system model under consideration in this correspondence is illustrated in Fig. \ref{fig:sm}. The system consists of a primary transmitter, a primary receiver, a passive eavesdropper, a secondary transmitter, and a secondary receiver. The secondary transmitter has multiple antennas, $N_t$, and all other nodes have a single antenna.\footnote{In \cite{Wu11TIFS}, the case of single antenna at the secondary transmitter was studied with a similar concept. When the secondary transmitter has a single antenna, the strategy at the secondary transmitter is limited to the power adjustment. In our work, we assume that the secondary transmitter has multiple antennas, i.e., $N_t \ge 2$ and study a design of an optimal beamforming vector.}
We assume that channel gains directly connected to nodes are available by exploiting the reciprocity of channels. For example, the secondary transmitter perfectly knows $\pmb{h}_{sp}, \pmb{h}_{se}, \pmb{h}_{ss}$ and the receivers have the relevant receiver-side channel state information. Moreover, the secondary transmitter uses a beamforming strategy, i.e., a scalar coding strategy that results in a unit-rank input covariance matrix. In this model, we study a design of the beamforming vector at the secondary transmitter which aims at improving the security of the primary network. 

The problem of spectrum leasing via cooperation for enhanced physical-layer secrecy can be formulated as the maximization of the secrecy rate of the primary network subject to the power constraint and the quality of service (QoS) constraint of the secondary network. The latter is given by imposing that the achievable rate of the secondary network is larger than a given threshold. If the power constraints are given by $P_p$ for the primary link and $\Psmax$ for the secondary link and $\pmb{w}$ denotes the beamforming vector at the secondary transmitter, the problem of the proposed idea corresponds to the following formulation:
\begin{equation}
\begin{array}{*{20}c}\label{PB1}
   {\mathop {\max }\limits_{{\pmb{w}}} } & {\Rsec \left(\pmb{w}\right)}  \\
   {\rm{s.t.}} & {\left\| {{\pmb{w}} } \right\|^2 \le \Psmax} \\
                &~~~ \Rss\left(\pmb{w}\right) \ge R_{\min},
 \end{array}
\end{equation}
where $R_{\min}$ is the required rate given to the secondary user in exchange for cooperation. The achievable rate for the secondary link is given by
\begin{equation}\label{Rss}
R_{ss}\left(\pmb{w}\right) = \log \left( 1 + \frac{\left|
\pmb{w}^*{\pmb{h}}_{ss} \right|^2 } {\sigma_s^2  + \left| {h}_{ps}
\right|^2P_p } \right),
\end{equation}
where $\sigma_s^2$ is the noise variances at the secondary receiver. To ensure feasibility of (\ref{PB1}), one can set
\begin{equation}\label{Rmin}
R_{\min} = \alpha R_{s, \max},
\end{equation}
where $\alpha \in [0, 1]$ and $R_{s, \max}=\Rss\left(\sqrt{\Psmax}\pmb{h}_{ss}/\|\pmb{h}_{ss}\|\right) = \log \left( 1 + \frac{\Psmax\left\|{\pmb{h}}_{ss} \right\|^2 }{\sigma_s^2  + \left|{h}_{ps} \right|^2P_p } \right)$. In fact, $R_{s, \max}$ is the maximum achievable rate and corresponds to a case where the beamforming vector is chosen to maximize the secondary rate, namely as the maximum ratio transmission (MRT), matched to the channel $\pmb{h}_{ss}$. Therefore, a parameter $\alpha$ represents the QoS level requested by the secondary network from the lowest ($\alpha=0$) to the highest ($\alpha=1$). Note that $\Rsec(\pw)$ in (\ref{PB1}) can be different depending upon the operation of the eavesdropper. In this correspondence, we consider two scenarios: a single-user decoding eavesdropper scenario and a joint decoding eavesdropper scenario. In the single-user decoding scenario, we assume that the eavesdropper treat the interference from the secondary transmitter as an additive noise. We also study the case of performing a joint decoding at the eavesdropper where the eavesdropper tries to decode messages from both the primary transmitter and the secondary transmitter.

\section{Single-user Decoding at the Eavesdropper}\label{sec_sd}
In this section, we discuss an optimal solution of the problem (\ref{PB1}) when the intended receiver and the eavesdropper treat the interference from the secondary transmitter as an additive noise. Given the assumptions, the following rate is achievable by the primary link with perfect secrecy 
\begin{equation}\label{eq_Rsec_sdsd}
\Rsec\left(\pmb{w}\right) = \left[\log \left( 1 + \frac{\left| {h}_{pp} \right|^2P_p } {\sigma_p^2  + \left| {\pmb{w}} ^* {\pmb{h}}_{sp} \right|^2 } \right) - \log \left( {1 + \frac{{\left| {h_{pe} } \right|^2P_p }}{{\sigma_e^2  + \left| {{\pmb{w}} ^* {\pmb{h}}_{se} } \right|^2 }}} \right)\right]^+,
\end{equation}
where $\pmb{w}$ denotes a beamforming vector at the secondary transmitter, $h_{ij}$ or $\pmb{h}_{ij}$ are the channel coefficient or $N_t \times 1$ channel vector between nodes, and $\sigma_p^2$ and $\sigma_e^2$ are the noise variances at the primary receiver and the eavesdropper, respectively. 

\subsection{Power Gain Region}

In this part, we review the concept and main results related to a power gain region as introduced in \cite{Jorswieck11TSP}. In particular, we introduce a design of beamforming vectors to achieve Pareto boundary points of the power gain region. Assume that there are a single transmitter with $N$ antennas and $K$ receivers with a single antenna. For beamforming transmission strategies, an achieved power gain at the $k$-th receiver is defined as
\begin{equation}
x_{k}(\pmb{w}) = \left|\pmb{w}^*\pmb{h}_{k}\right|^2,
\end{equation}
where $\pmb{w}$ is a beamforming vector at the transmitter and $\pmb{h}_{k}$ is a channel vector between the transmitter and the $k$-th receiver, $k \in \mathcal{K}$, $\mathcal{K} = \left\{1, 2, \ldots, K \right\}$. Then, a power gain region with a transmit power constraint, i.e., $\left\|\pmb{w}\right\|^2 \le \Pmax$, is defined as a set of all achievable power gains as follows:
\begin{equation}
\Omega = \left\{\pmb{x}(\pw) = \left(x_{1}(\pmb{w}), x_{2}(\pmb{w}), \ldots, x_{K}(\pmb{w})\right)| \left\|\pmb{w}\right\|^2 \le \Pmax \right\},
\end{equation}
and, given $\pmb{e}$ where $e_i \in \{-1, +1\}$, the outer boundary of the power gain region in direction $\pmb{e}$ is defined as follows:
\begin{equation}\label{eq_outbound}
\mathcal{B}^{\pmb{e}}\Omega = \left\{\pmb{x}'| \pmb{x}'\ge^{\pmb{e}} \pmb{x} 
, \forall \pmb{x} \in \Omega, \pmb{x}' \in \Omega \right\},
\end{equation}
where $\pmb{a} \ge^{\pmb{e}} \pmb{b}$ if $a_ie_i \ge b_ie_i$ and $\pmb{x} = \left[x_{1}(\pmb{w}), x_{2}(\pmb{w}), \ldots, x_{K}(\pmb{w}) \right]^T$. In Fig.~\ref{fig_PGR}, we give an example to show the power gain region when $(N, K) = (2, 2)$. Depending on $\pmb{e}$, there are three different power gain regions. In particular, if we set $\pmb{e} = \left[+1 +1\right]$, we get $\mathcal{B}^{\pmb{e}}\Omega$ that is shown in the upper right part of the boundary. When $\pmb{e} = \left[+1 -1\right]$, $\mathcal{B}^{\pmb{e}}\Omega$ is drawn as the right-lower part and likewise for $\pmb{e} = \left[-1 +1\right]$.
For a design of $\pmb{e}$, $e_i$ is set to $+1$ if the $i$-th receiver is an intended receiver, while $e_j = -1$ if the $j$-th receiver is an unintended receiver. By doing so, $\mathcal{B}^{\pmb{e}}\Omega$ becomes a set of Pareto optimal points of the power gain region in direction to maximize its power gain at intended receivers and/or minimizes its power gain at unintended receiver \cite{Jorswieck11TSP}.

\begin{lemma}\label{lemma_pg1}
On the assumption that channels are linearly independent, all the points $\pmb{x}(\pmb{w})$ on the outer boundary of $\Omega$ in direction $\pmb{e}$, $\mathcal{B}^{\pmb{e}}\Omega$, can be achieved by using $\pw$ as follows:
\begin{equation}\label{eq_pg_bound}
\pw = \sqrt{P} \pmb{v}_{\max}\left\{\pZ\right\},
\end{equation}
where $\vmax(\pZ)$ is the eigenvector with unit norm corresponding to the maximum eigenvalue of $\pZ$ and 
\begin{equation}
\pmb{Z} = \sum\limits_{k = 1}^{K} \mu_{k}e_{k}\pmb{h}_{k}\pmb{h}_{k}^*,
\end{equation}
for some $\mu_{k}$ such that $\mu_{k} \in \left[0, 1\right], \sum\limits_{k=1}^{K} \mu_{k} = 1$,
where $P$ is chosen as follows:
\begin{equation}\label{eq_pg_bound_power}
P = 
\left\{ {\begin{array}{ll}
   {\Pmax } &    {, \lambda_{\max}(\pmb{Z}) > 0 \;\text{or}\; N \ge K},  \\
   {[0 \;\; \Pmax ]}  &    {, \lambda_{\max}(\pmb{Z}) = 0}, \\
   0  &    {, \lambda_{\max}(\pmb{Z}) < 0}.  \\
 \end{array} } \right.
\end{equation}
\end{lemma}
\begin{proof}
The detailed proof of \textit{Lemma} \ref{lemma_pg1} is shown in \cite{Jorswieck11TSP}.
\end{proof}

Note that, if $N \ge K$, all boundary points are achieved only by $P=\Pmax$. In this case, if the transmitter has an extra power, the extra power can be used to increase power gains at some intended receivers or decrease power gains at some unintended receivers. The points marked by star in Fig.~\ref{fig_PGR} are examples of this situation.
On the other hand, when $N < K$, $\pZ$ may have $\lambda_{\max}(\pZ) \le 0$ and some boundary points are achieved by power adjustment. For instance, if $(N, K) = (3, 4)$ and there are three unintended receivers and one intended receiver, the beamforming vector to enforce zero power gains to two unintended receivers has only one degree of freedom. At the same time, the beamforming vector has two conflict goals such that decreasing a power gain at the last unintended receiver and increasing a power gain at the intended receiver. In this case, the only way to achieve outer boundary points is power adjustment.

\subsection{A Design of an Optimal Beamforming Vector}

In this part, we study a design of an optimal beamforming vector as a solution of problem (\ref{PB1}) with $\Rsec$ given in (\ref{eq_Rsec_sdsd}). Suppose that $P_s$ indicates a used power at the secondary transmitter and the power constraint at the secondary transmitter is given as $\Psmax$, i.e., $\|\pw\|^2 \le \Psmax$.

\begin{proposition}\label{proposition_sdsd}
On the assumption that channels are linearly independent, the optimization problem (\ref{PB1}) with $\Rsec$ given in (\ref{eq_Rsec_sdsd}) can be solved by one of elements in the set as follows:
\begin{equation}\label{eq_prop_sdsd}
\pmb{w}_{\rm{opt}} \in \left\{ \pw \left| \pw = \sqrt{\Psmax} \pmb{v}_{\max}\left\{\pmb{Z}\right\} \right., \pmb{Z} = -\mu_1\hsp\hsp^* +\mu_2\hse\hse^* + \mu_3\hs\hs^*, \mu_{k} \in \left[0, 1\right], \sum_{k=1}^{3} \mu_k = 1\right\}.
\end{equation}
\end{proposition}

\begin{proof}
Based on \textit{Lemma} \ref{lemma_pg1}, we see that (\ref{eq_prop_sdsd}) presents the outer boundary points of power gain region, $(|\pw^*\hsp|^2, |\pw^*\hse|^2, |\pw^*\hs|^2),$ in direction $\pe_1 = [ -1 \; +1 \; +1]$ which are obtained by $P_s = \Psmax$. Specifically, \textit{Proposition} \ref{proposition_sdsd} implies that $\wopt$ must exist on the set of boundary points where the primary receiver is an unintended receiver, while the eavesdropper and the secondary receiver are intended receivers. By contradiction, we first prove that $\wopt$ exists on the outer boundary of the power gain region. Assume that the optimal beamforming vector $\pw_1$ is not on the boundary points of the power gain region in direction $\pe_1$. Then, by definition of the boundary points of the power gain region (\ref{eq_outbound}), we can find another beamforming vector $\pw_2$ in the boundary points in direction $\pe_1$ that satisfies at least one of the following three cases.
\begin{enumerate}
\item 
$\left| {\pw_1}^*\hsp \right|^2 >
\left| {\pw_2}^* \hsp \right|^2$,
$\left| {\pw_1}^*\hse \right|^2 =
\left| {\pw_2}^* \hse \right|^2$,
$\left| {\pw_1}^*\hs \right|^2 =
\left| {\pw_2}^* \hs \right|^2$,
\item 
$\left| {\pw_1}^*\hsp \right|^2 =
\left| {\pw_2}^* \hsp \right|^2$,
$\left| {\pw_1}^*\hse \right|^2 <
\left| {\pw_2}^* \hse \right|^2$,
$\left| {\pw_1}^*\hs \right|^2 =
\left| {\pw_2}^* \hs \right|^2$,
\item
$\left| {\pw_1}^*\hsp \right|^2 =
\left| {\pw_2}^* \hsp \right|^2$,
$\left| {\pw_1}^*\hse \right|^2 =
\left| {\pw_2}^* \hse \right|^2$,
$\left| {\pw_1}^*\hs \right|^2 <
\left| {\pw_2}^* \hs \right|^2$,
\end{enumerate}
According to above results, we can find $\pw_2$ on the boundary points of the power gain region in direction $\pe_1$ such that $\Rsec(\pw_2) \ge \Rsec(\pw_1)$ and $\Rss(\pw_2) \ge \Rss(\pw_1)$ for any $\pw_1$, which is not on the boundary points. Therefore, the proof showing the existence of $\wopt$ on the boundary can be concluded.

Next, we show that $\wopt$ is obtained by $P_s = \Psmax$, irrespective of $N_t$. In \textit{Lemma} \ref{lemma_pg1}, it was shown that all outer boundary points are achieved with $P = \Pmax$ if $N \ge K$. In our system model, $N = N_t$ and $K=3$; thus, when $N_t \ge 3$, all the boundary points are achieved by $P_s = \Psmax$. If $N_t = 2$, the outer boundary points of the power gain region include the points achieved by a power adjustment. It is worth noting that, when $\pZ$ is given as (\ref{eq_prop_sdsd}), $\lambda_{\max}(\pZ) \le 0 $ occurs only if $\pZ = -\hsp\hsp^*$. If $\pZ = -\hsp\hsp^*$, $\lambda_{\max}(\pZ)$ is equal to zero and the corresponding eigenvector satisfies $\hsp^*\vmax(\pZ) = 0$. Then, the boundary points which are achieved by $\pw = \sqrt{P_s}\vmax(\pZ)$ can be written as
\begin{equation}\label{eq:pgr3}
\left(|\pw^*\hsp|^2, |\pw^*\hse|^2, |\pw^*\hs|^2 \right) = \left(0, P_s|\vmax(\pZ)^*\hse|^2, P_s|\vmax(\pZ)^*\hs|^2 \right)
\end{equation}
where $P_s \in [0 \; \Psmax]$. Interestingly, putting (\ref{eq:pgr3}) into (\ref{eq_Rsec_sdsd}) shows that $\Rsec(\pw)$ and $\Rss(\pw)$ are maximized if $P = \Psmax$. Therefore, $\wopt$ is also obtained when $P_s = \Psmax$ if $N_t = 2$.
\end{proof}

It is observed that, in \textit{Proposition} \ref{proposition_sdsd}, (\ref{eq_prop_sdsd}) includes three real numbers, but they can be obtained by the combinations of two real numbers since the sum of three numbers needs to be unity.

\section{Joint Decoding at the Eavesdropper}\label{sec_jd}

In Section \ref{sec_sd}, we assumed that the eavesdropper treated the interference from the secondary transmitter as an additive noise. In this section, we study a case where the eavesdropper performs a joint decoding. Since the eavesdropper receives a superposition of signals from the primary transmitter and from the secondary transmitter, the eavesdropper can be intelligent enough to decode messages from both the primary transmitter and the secondary transmitter. This is the worst scenario in terms of the secrecy rate and provides a lower bound of the achievable performance of the proposed idea. We investigate an achievable secrecy rate when the eavesdropper performs a joint decoding and propose a design of an optimal beamforming vector.

\subsection{Achievable Secrecy Rate with a Joint Decoding Eavesdropper}

\begin{lemma}\label{lem_rr_sdjd}
Let $R_p$ and $R_s$ denote a rate of the message at the primary transmitter and the secondary transmitter, respectively. Then, the following rate region $(R_p, R_s)$ is achievable at the eavesdropper:
\begin{equation}
\mathcal{R}^{\rm{Eve}} = \{(R_p, R_s)| (R_p, R_s) \in \mathcal{R}_{\rm{MAC}}^{\rm{Eve}} \cup \mathcal{R}_{\rm{SD}}^{\rm{Eve}}\}
\end{equation}
where
\begin{equation}
\mathcal{R}_{\rm{MAC}}^{\rm{Eve}} = 
\left\{ (R_p, R_s) \left| 
{\begin{array}{*{20}c}
 R_p \le \RpeJD  \\
 R_s \le \RseJD  \\
 R_p + R_s \le \ReMAC  \\
\end{array}} \right. \right\},
\end{equation}
\begin{equation}
\mathcal{R}_{\rm{SD}}^{\rm{Eve}} = \left\{ (R_p, R_s) |  R_p \le \RpeSD \right\},
\end{equation}
\begin{align*}
&\RpeJD = \log \left( 1 + \frac{\left| \hpe \right|^2 P_p} {\sigmae } \right),
&&\RseJD = \log \left( 1 + \frac{\left| \pw^*\hse \right|^2} {\sigmae } \right),\\
&\ReMAC = \log \left( 1 + \frac{\left| \hpe \right|^2 P_p + \left| \pw^*\hse \right|^2}{\sigmae} \right),
&&\RpeSD = \log \left( 1 + \frac{\left| \hpe \right|^2 P_p} {\sigmae  + \left| \pw^*\hse \right|^2} \right).
\end{align*}
\end{lemma}
\begin{proof}
The eavesdropper receives the messages from both the primary transmitter and the secondary transmitter. This setup is the same as multiple access channel and $\mathcal{R}_{\rm{MAC}}^{\rm{Eve}}$ is achievable at the eavesdropper. On the other hand, since the eavesdropper is interested only in the message from the primary transmitter, the eavesdropper may treat the received signal from the secondary transmitter as an additive noise while not decoding the message from the secondary transmitter. In this case, $\mathcal{R}_{\rm{SD}}^{\rm{Eve}}$ is achievable at the eavesdropper. Finally, the union of $\mathcal{R}_{\rm{MAC}}^{\rm{Eve}}$ and $\mathcal{R}_{\rm{SD}}^{\rm{Eve}}$ is achievable at the eavesdropper. For details, refer to \cite{Tang11TIT}.
\end{proof}

We assume that the primary receiver does not perform the joint decoding; thus the achievable rate region of the primary receiver is as follows:
\begin{equation}
\mathcal{R}^{\rm{PRx}} = \mathcal{R}_{\rm{SD}}^{\rm{PRx}} = \left\{(R_p, R_s) | R_p \le \RppSD \right\},
\end{equation}
where
\begin{equation*}
\RppSD =\log \left( 1 + \frac{\left| \hpp \right|^2 P_p} {\sigmap  + \left| \pw^*\hsp \right|^2} \right).
\end{equation*}
The achievable rate region of the primary receiver and the eavesdropper are shown in Fig. \ref{fig_rr_sdjd}.

\begin{lemma}\label{lem_sdjd}
Suppose that $\mathcal{R}^{\rm{PRx}}$ and $\mathcal{R}^{\rm{Eve}}$ are the achievable rate regions of the primary receiver and the eavesdropper, respectively. Given $\mathcal{R}^{\rm{PRx}}$ and $\mathcal{R}^{\rm{Eve}}$, following rates are achievable with perfect secrecy 
\begin{equation}\label{eq_rsec_jd}
\Rsec = \left\{R| R = [R_p^1 - R_p^2, 0]^+, (R_p^1, R_s) \in \mathcal{R}^{\rm{PRx}}, (R_p^2, R_s) \notin \mathcal{R}^{\rm{Eve}}	 \right\},
\end{equation}
\end{lemma}
\begin{proof}
The detailed proof is provided in \cite{Tang11TIT}\footnote{Note that \cite{Tang11TIT} studied an information theoretic analysis of wire-tap channel with a helping interferer. However, since the secondary transmitter works like a helping interferer in our system, the achievable rate with perfect secrecy can still be proved in the same way.}.
\end{proof}


Based on $\mathcal{R}^{\rm{PRx}}$, $\mathcal{R}^{\rm{Eve}}$ and \textit{Lemma} \ref{lem_sdjd}, the following rate is achievable with perfect secrecy:
\begin{equation}\label{eq_rsec_sdjd}
   \Rsec  = \left\{ {\begin{array}{ll}
   \RppSD - \RpeSD &, \RseJD \le \Rss\\
   \RppSD - \ReMAC + \Rss &, \RseSD \le \Rss < \RseJD\\
   \RppSD - \RpeJD &, \Rss < \RseSD
 \end{array} } \right.
\end{equation}
where
\begin{equation*}
\RseSD = \log \left( 1 + \frac{\left| \pw^*\hse \right|^2}  {\sigmae  + \left| \hpe \right|^2P_p} \right).
\end{equation*}
Suppose that $\wopt_i, i\in \{1, 2, 3\}$ are solutions of (\ref{PB1}) with $\Rsec$ given in each case of (\ref{eq_rsec_sdjd}) where additional constraints of corresponding $\Rss$ are added such as
\begin{align}
\hspace{-0.5cm}\wopt_1 &= \arg \mathop {\max }\limits_{\{\pw\}} \{\RppSD - \RpeSD\},  \hspace{0.5cm} {\rm{s.t.}} \; \left\| {{\pw} } \right\|^2 \le \Psmax, \;\; \Rss \ge R_{\min}, \;\; \RseJD \le \Rss,\label{eq_wopt1}\\
\hspace{-0.5cm}\wopt_2 &= \arg \mathop {\max }\limits_{\{\pw\}} \{\RppSD - \ReMAC + \Rss \},  \hspace{0.5cm} {\rm{s.t.}} \; \left\| {{\pw} } \right\|^2 \le \Psmax, \;\; \Rss \ge R_{\min}, \;\; \RseSD \le \Rss \le \RseJD,\label{eq_wopt2}\\
\hspace{-0.5cm}\wopt_3 &= \arg \mathop {\max }\limits_{\{\pw\}} \{\RppSD - \RpeJD\},  \hspace{0.5cm} {\rm{s.t.}} \; \left\| {{\pw} } \right\|^2 \le \Psmax, \;\; \Rss \ge R_{\min}, \;\; \Rss \le \RseSD.\label{eq_wopt3}
\end{align} 
Then, the optimization problem (\ref{PB1}) with $\Rsec$ given in (\ref{eq_rsec_sdjd}) can be obtained by choosing a vector $\pw$ that leads to the largest $\Rsec$ among $\wopt_1, \wopt_2,$ and $\wopt_3$. Note that the solution with a joint decoding eavesdropper is selected from a comparison of three candidates unlike the solution with single-user decoding eavesdropper.

\subsection{A Design of an Optimal Beamforming Vector}

For the solution of (\ref{PB1}) with $\Rsec$ given in (\ref{eq_rsec_sdjd}), we need to find $\wopt_1, \wopt_2,$ and $\wopt_3$. Regarding $\wopt_1$, (\ref{eq_wopt1}) shows that the optimization problem is equivalent to the problem of \textit{Proposition} \ref{proposition_sdsd} with an additional constraint $\RseJD \le \Rss$. The additional constraint is rewritten as
\begin{equation}
\frac{\sigmas + \left|\hps\right|^2P_p}{\sigmae} \le \frac{\left|\pw^*\hs \right|^2}{\left| \pw^*\hse \right|^2}.
\end{equation}
To find $\wopt_1$, 
we use the solution provided in \textit{Proposition} \ref{proposition_sdsd} by adding the constraint numerically. Suppose that the set in \textit{Proposition} \ref{proposition_sdsd} is defined as
\begin{equation}\label{eq_prop_sdjd11}
\mathcal{S}_{1} = \left\{ \pw \left| \pw = \sqrt{\Psmax} \vmax\left\{\pZ\right\} \right., \pZ = -\mu_1\hsp\hsp^* +\mu_2\hse\hse^* + \mu_3\hs\hs^*, \mu_{k} \in \left[0, 1\right], \sum_{k=1}^{3} \mu_k = 1\right\}.
\end{equation}
Then, $\wopt_1$ can be found on the set as follows:
\begin{equation}\label{eq_prop_sdjd12}
\wopt_1 \in \left\{ \pw \left| \pw \in \mathcal{S}_{1}, \frac{\sigmas + \left|\hps\right|^2P_p}{\sigmae} \le \frac{\left|\pw^*\hs \right|^2}{\left| \pw^*\hse \right|^2} \right. \right\}.
\end{equation}
For finding $\wopt_2$, since $\Rsec(\pw)$ for $\wopt_2$ is different from that for $\wopt_1$, we establish a new solution as follows:
\begin{proposition}\label{proposition_sdjd}
On the assumption that channels are linearly independent, $\wopt_2$ can be found on the set as follows:
\begin{equation}\label{eq_prop_sdjd23}
\wopt_2 \in \left\{ \pw \left| \pw \in \mathcal{S}_{2}, \;\; \frac{\sigmas + \left|\hps\right|^2P_p}{\sigmae  + \left| \hpe \right|^2P_p} \le \frac{\left|\pw^*\hs \right|^2}{\left| \pw^*\hse \right|^2} \le \frac{\sigmas + \left|\hps\right|^2P_p}{\sigmae} \right. \right\}
\end{equation}
where $\mathcal{S}_2$ represents the solution of the optimization problem (\ref{PB1}) with $\Rsec = \RppSD - \ReMAC + \Rss$ and it can be acquired on the set as follows:
\begin{equation}\label{eq_prop_sdjd21}
\mathcal{S}_{2} = \left\{ \pw \left| \sqrt{P_s} \vmax\{\pZ\} \right., \pZ = - \mu_1\hsp\hsp^* - \mu_2\hse\hse^* + \mu_3\hs\hs^*, \mu_{k} \in \left[0, 	1\right], \sum\limits_{k=1}^{3} \mu_k = 1\right\},
\end{equation}
where  $P_s$ is chosen as follows:
\begin{equation}\label{eq_prop_sdjd22}
P_s = 
\left\{ {\begin{array}{*{20}l}
   {\Psmax } &    {, \lambda_{\max}(\pmb{Z}) > 0 \; \text{or} \; N_t \ge 3},  \\
   {[0 \;\; \Psmax ]}  &    {, \lambda_{\max}(\pmb{Z}) = 0}, \\
   0  &    {, \lambda_{\max}(\pmb{Z}) < 0}.  \\
 \end{array} } \right.
\end{equation}
\end{proposition}

\begin{proof}
Based on \textit{Lemma} \ref{lemma_pg1}, we observe that $\mathcal{S}_2$ in (\ref{eq_prop_sdjd21}) presents the outer boundary points of the power gain region in direction $\pe_2 = [ -1 \; -1 \; +1]$. The proof is similar to that in \textit{Proposition} \ref{proposition_sdsd}. Briefly, if the optimal beamforming vector $\pw_1$ is not on the boundary points of the power gain region in direction $\pe_2$, we can find another beamforming vector $\pw_2$ on the boundary points in direction $\pe_2$ that has $\Rsec(\pw_1) \ge \Rsec(\pw_2)$ or $\Rss(\pw_1) \ge \Rss(\pw_2)$ when $\Rsec = \RppSD - \ReMAC + \Rss$. Note that, unlike \textit{Proposition} \ref{proposition_sdsd} where $P_s = \Psmax$, $\mathcal{S}_2$ includes the boundary points achieved a power adjustment, i.e., $P_s \in [0 \; \Psmax]$. Finally, given $\mathcal{S}_2$, $\wopt_2$ can be acquired in $\mathcal{S}_2$ by adding the constraint as in (\ref{eq_prop_sdjd23}). 
\end{proof}

Note that the case of $\lambda_{\max}(\pZ) \le 0$ occurs when $N_t < 3$ in (\ref{eq_prop_sdjd22}).

\begin{remark}
On the assumption that channels are linearly independent, $\wopt_3$ can be found on the set as follows:
\begin{equation}\label{eq_prop_sdjd32}
\wopt_3 \in \left\{ \pw \left| \pw \in \mathcal{S}_{3}, \;\; \frac{\left|\pw^*\hs \right|^2}{\left| \pw^*\hse \right|^2} \le \frac{\sigmas + \left|\hps\right|^2P_p}{\sigmae  + \left| \hpe \right|^2P_p} \right. \right\},
\end{equation}
where $\mathcal{S}_3$ represents the solution of the optimization problem (\ref{PB1}) with $\Rsec = \RppSD - \RpeJD$ and it can be acquired on the set as follows:
\begin{equation}\label{eq_prop_sdjd31}
\mathcal{S}_3 = \left\{ \pw \left| \sqrt{P_s} \vmax\{\pZ\} \right., \pZ = - \mu_1\hsp\hsp^* + \mu_2\hs\hs^*, \mu_{k} \in \left[0, 1\right], \sum_{k=1}^{2} \mu_k = 1\right\}.
\end{equation}
where $P_s \in [0 \; \Psmax]$.
\end{remark}

In (\ref{eq_prop_sdjd21}) and (\ref{eq_prop_sdjd31}), $P_s \in [0 \; \Psmax]$ implies that the optimal beamforming vector may be achieved with power adjustment. Basically, in our model, the secondary transmitter aims at interrupting the  eavesdropper. When the eavesdropper performs a joint decoding, the interference at the eavesdropper might not be effective due to the joint decoding of the eavesdropper while the primary receiver still suffers from the interference by the secondary transmission. Therefore, decreasing the transmit power at the secondary transmitter is occasionally a better strategy, as shown in (\ref{eq_prop_sdjd21}) and (\ref{eq_prop_sdjd31}).


\section{Numerical Results}

In this section, we compare the proposed idea with the case of no eavesdropper, i.e., peaceful system (upper bound), and with the case of no spectrum leasing as a function of $P_p/\sigma_p^2 = P_s/\sigma_s^2$, which, hereinafter, is defined as signal to noise ratio (SNR). For the simulation, we use $0$, $0.5$, and $0.8$ for the required rate constraint, $\alpha$, in (\ref{Rmin}). Note that the case of $\alpha = 0$ means that the secondary transmitter can focus its beam solely on maximizing the secrecy rate and this case can be referred to a helping interferer \cite{Jorswieck10SCC}.

In Figs.~\ref{fig_perf_SDvsJD_SNR_Nt3} and~\ref{fig_perf_SDvsJD_SNR_Nt2}, we show the achievable secrecy rate of the proposed spectrum leasing when $N_t$ is equal to 3 and 2, respectively. In each graph, we compare the cases of different operation of decoding with various $\alpha$.
With a single-user decoding eavesdropper, it is shown that one can reap most of the benefits of the spectrum leasing while still serving the needs of the secondary network. Specifically, we observe that even when $\alpha = 0.5$, the secrecy rate with spectrum leasing is comparable, especially in the low SNR regime. Obviously, the secrecy rate decreases as $\alpha$ increases, and the performance increases as the number of antennas increases. When the eavesdropper performs a joint decoding, the achievable secrecy rate with the proposed spectrum leasing drops compared to the case with a single-user decoding eavesdropper. It is of interest to see that most cases of the proposed idea outperform the case of no spectrum leasing even when the eavesdropper performs a joint decoding. In Fig.~\ref{fig_perf_SDvsJD_SNR_Nt2}, when $\alpha$ is equal to 0.8 and $N_t$ is equal to two, the proposed spectrum leasing fails to outperform the case of no spectrum leasing; indeed, the secondary rate requirement is too high and the number of antennas is insufficient to improve both the primary network and the secondary network. In this case, the cooperation of the primary and secondary transmitter could improve the performance by decreasing $\alpha$. We leave the issue of control of $\alpha$ for future work.

In Fig. \ref{fig_perf_SDvsJD_Nt}, we show the achievable secrecy rate of the proposed spectrum leasing as a function of the number of antennas at the secondary transmitter. It is observed that the performance with larger $N_t$ shows better performance than that with smaller $N_t$. With a single-user decoding eavesdropper, the proposed technique rapidly approaches the rate of a peaceful system. In particular, we observe that, numerically, the proposed idea with more than ten antennas shows a close performance of a peaceful system. When the eavesdropper performs a joint decoding, the achievable secrecy rate is also increased as the number of antennas increases. In this case, however, the performance approaches the performance of a peaceful system more slowly than the case with a single-user decoding eavesdropper. Moreover, the results show that the performance gap between different secondary rate constraints becomes marginal as the number of antennas increases. 

\section{Conclusion} \label{conclusion}

In this correspondence, we studied a new application of a spectrum leasing via cooperation in which the secondary transmission aims at improving the secrecy rate of the primary network. In particular, based on the framework of the power gain region, we proposed an optimal beamforming vector that maximizes the secrecy rate of the primary network while maintaining the rate constraint given to the secondary link. To provide a worst case scenario in terms of secrecy rate, we also investigated the case of joint decoding eavesdropper and it was shown the the proposed idea is still useful. Our work can be extended to designs of cognitive radio networks for physical-layer security in various ways.

\bibliographystyle{IEEEtran}
\bibliography{LKK_KAIST_PhD}
\newpage

\begin{figure}[h]
\centering
\includegraphics[width=15cm]{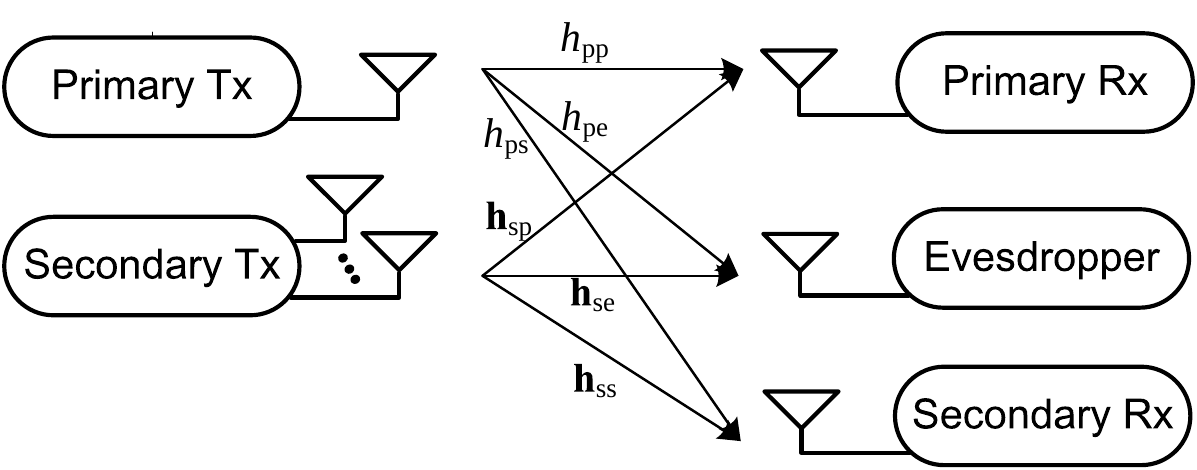}
\caption{System model with proposed spectrum leasing} \label{fig:sm}
\end{figure}

\begin{figure}[h]
\centering
\includegraphics[width=14cm]{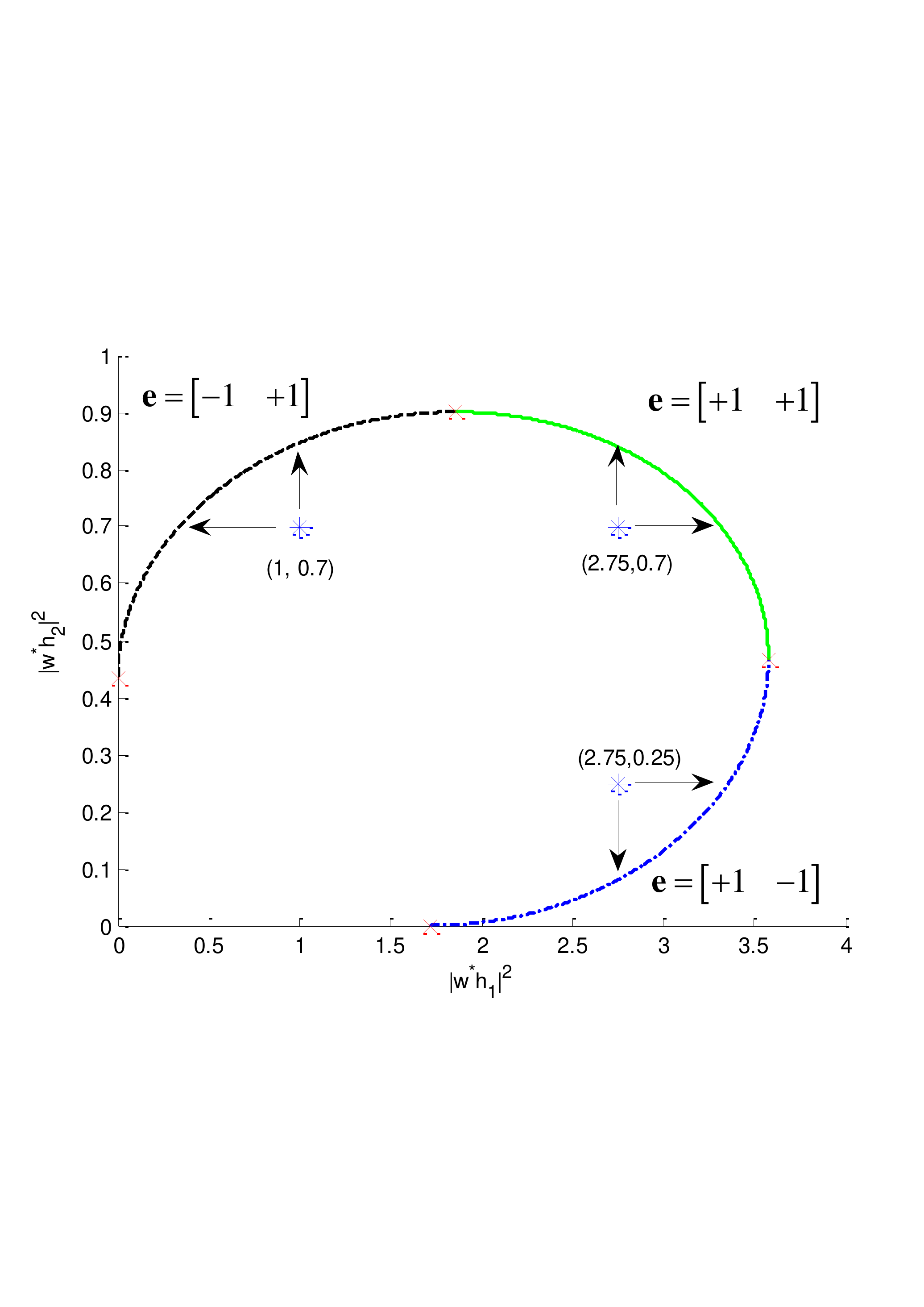}
\caption{An example of power gain region. $(N, K) = (2, 2)$.} \label{fig_PGR}
\end{figure}

\newpage

\begin{figure}[h]
\centering
\includegraphics[width=12cm]{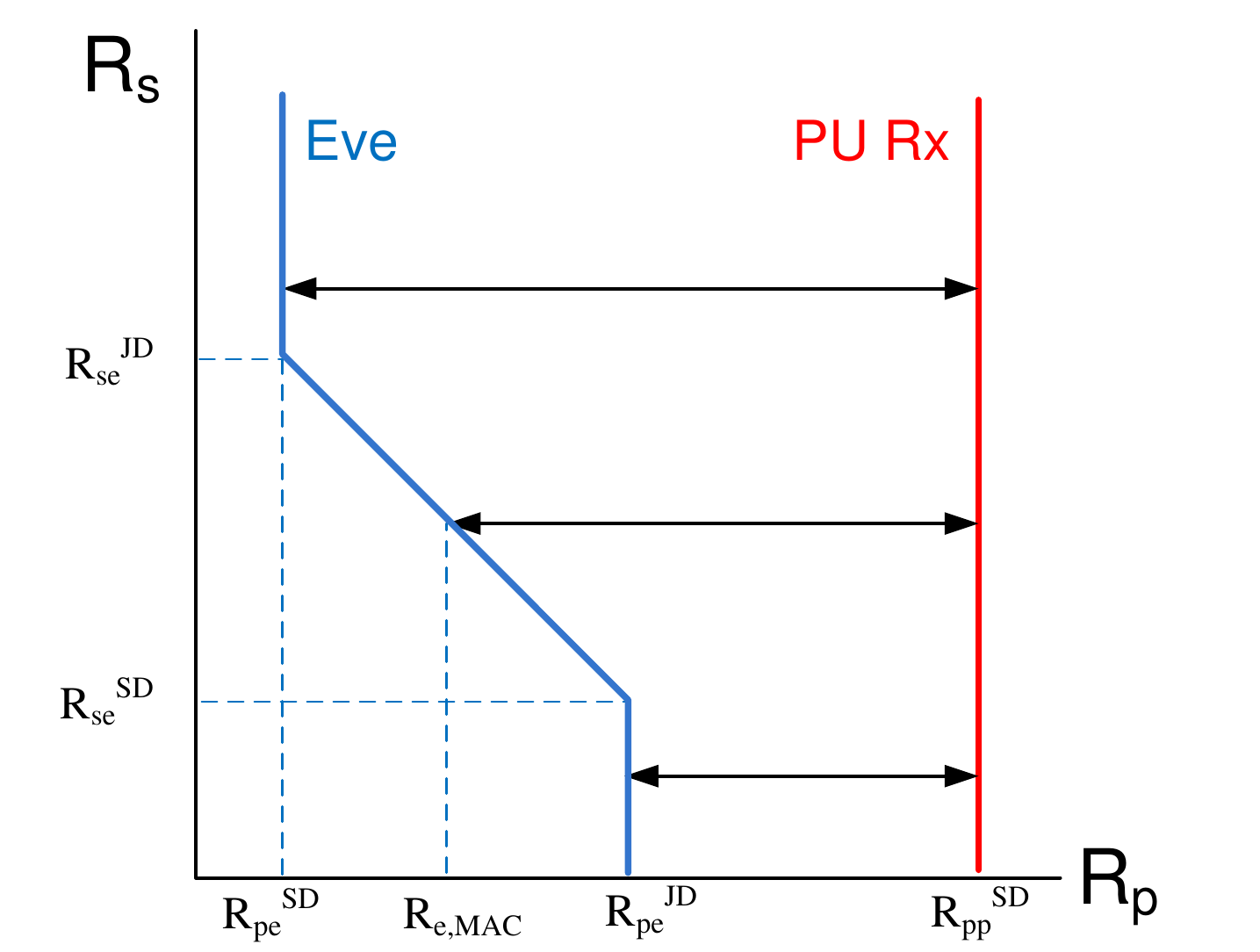}
\caption{Achievable rate region of the primary receiver and the eavesdropper in joint decoding scenario.} \label{fig_rr_sdjd}
\end{figure}

\begin{figure}[h]
\centering
\includegraphics[width=12cm]{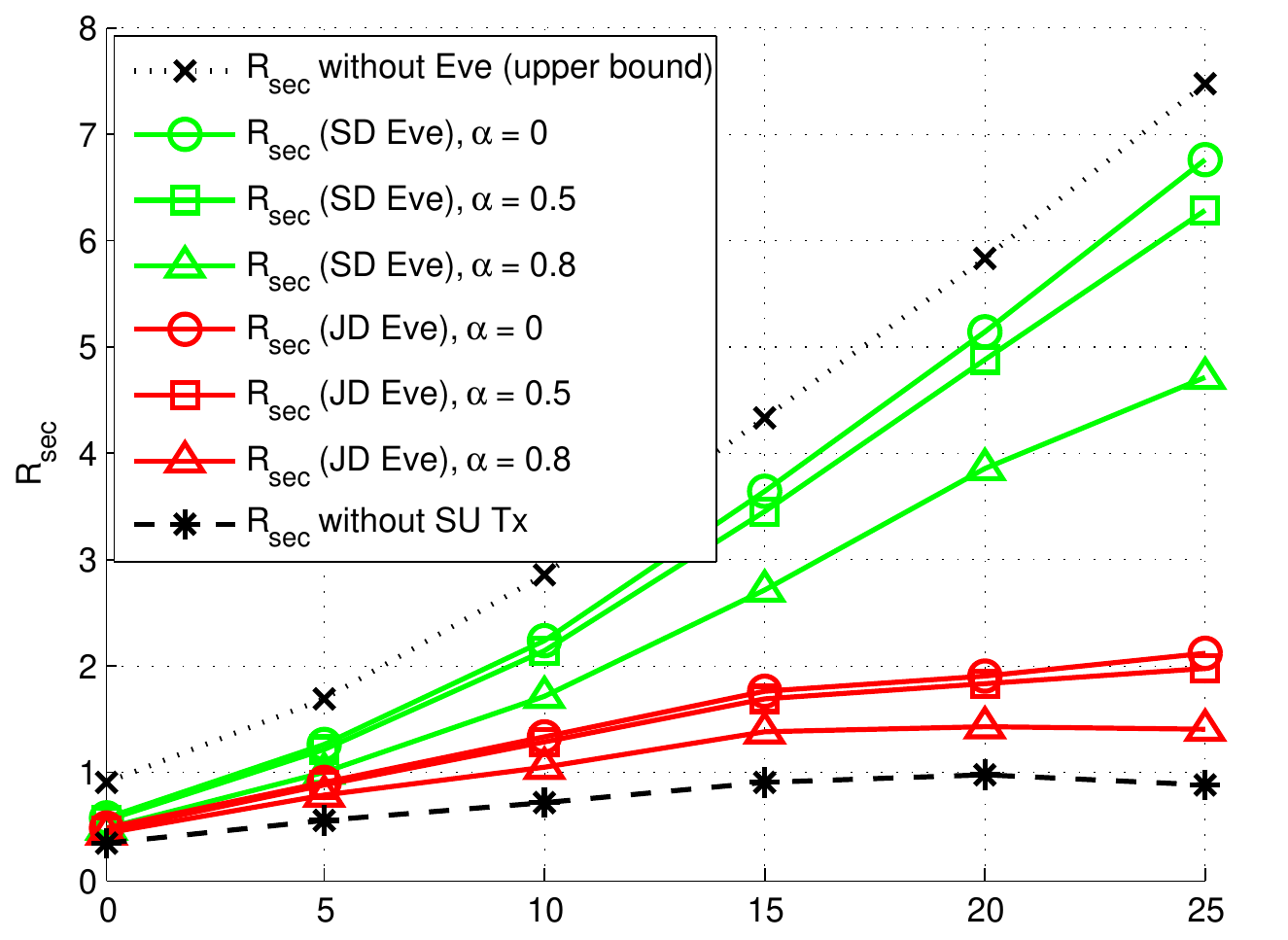}
\caption{Secrecy rate with single-user/joint decoding (SD/JD) eavesdropper when $N_t = 3$.} \label{fig_perf_SDvsJD_SNR_Nt3}
\end{figure}

\begin{figure}[h]
\centering
\includegraphics[width=13cm]{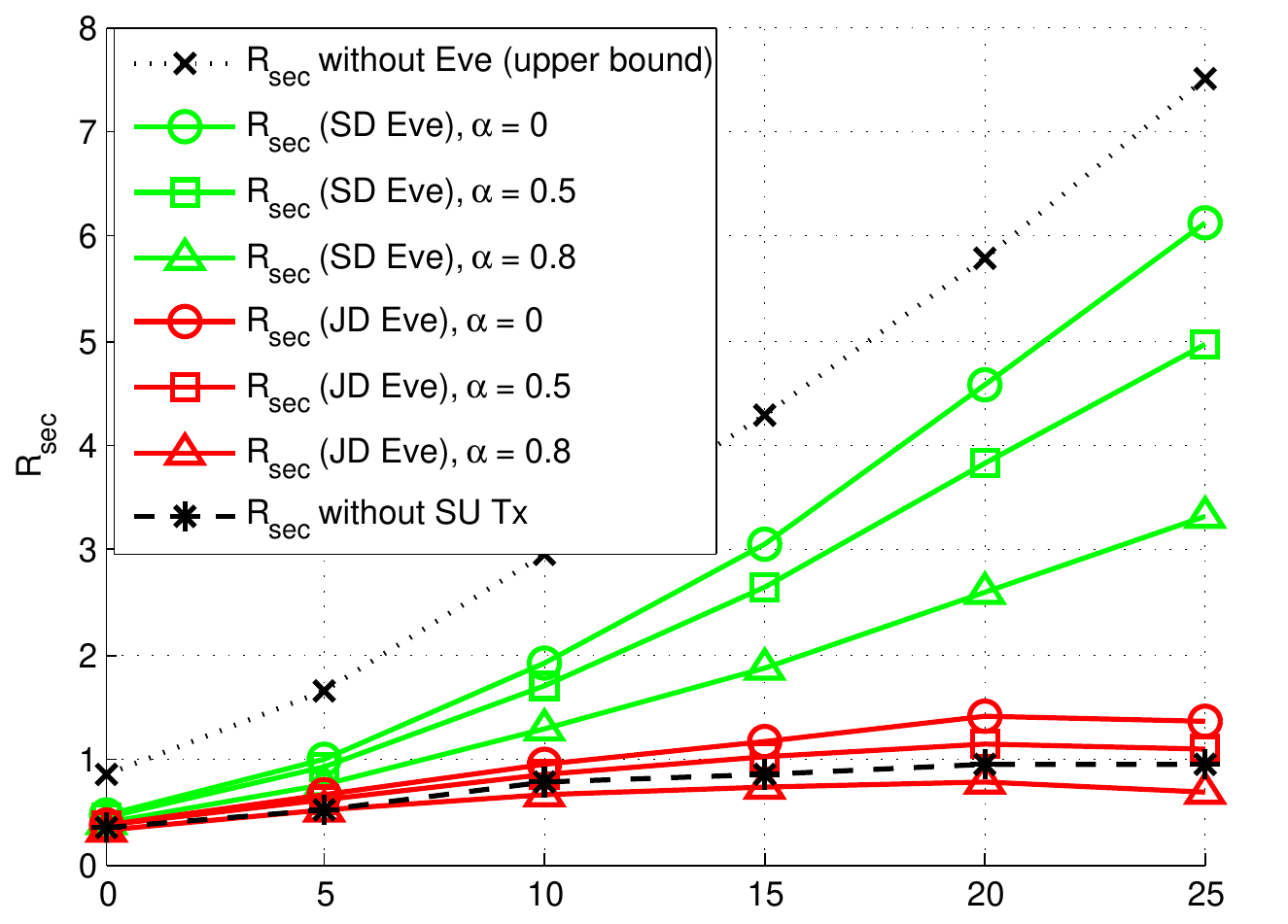}
\caption{Secrecy rate with single-user/joint decoding (SD/JD) eavesdropper when $N_t = 2$.} \label{fig_perf_SDvsJD_SNR_Nt2}
\end{figure}

\begin{figure}[h]
\centering
\includegraphics[width=13cm]{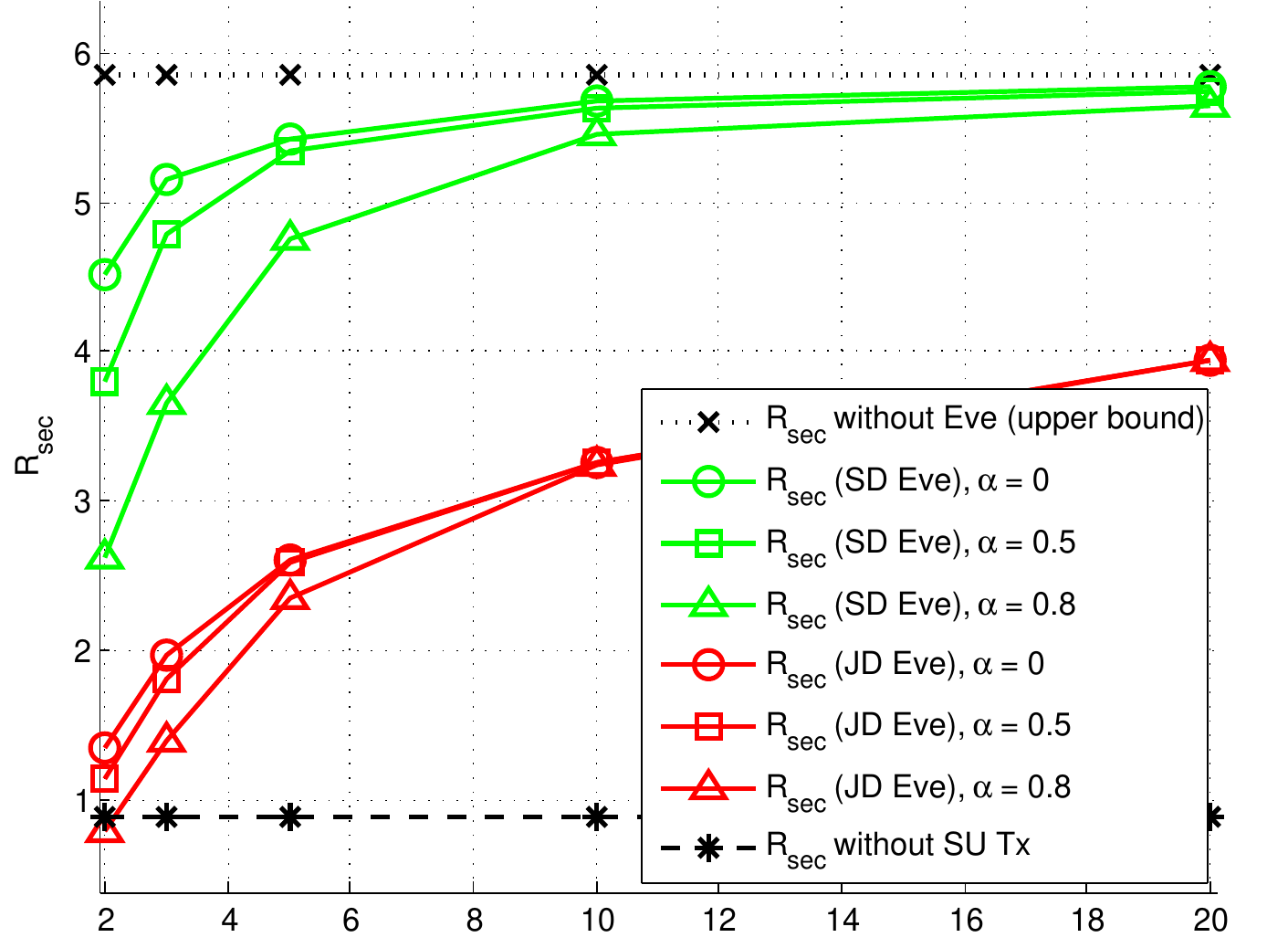}
\caption{Secrecy rate with single-user/joint decoding (SD/JD) eavesdropper when SNR $= 20$(dB).} \label{fig_perf_SDvsJD_Nt}
\end{figure}

\end{document}